\documentclass[a4p, 12pt]{article}
\title{Characterization of  degenerate   supersymmetric ground states
of the  Nicolai supersymmetric fermion lattice model by  symmetry breakdown} 
\date{August 21  2020}
\author{Hosho  Katsura,\thanks{Department of Physics, Graduate School of Science, The University of Tokyo, 7-3-1, Hongo,
Bunkyo-ku, Tokyo, 113-0033, Japan}$~^,$\thanks{Institute for Physics of Intelligence, The University of Tokyo, 7-3-1, Hongo, Bunkyo-ku, Tokyo,
113-0033, Japan}$~^,$\thanks{Trans-scale Quantum Science Institute, The University of Tokyo, 7-3-1, Hongo, Bunkyo-ku,
Tokyo, 113-0033, Japan}
~Hajime  Moriya,\thanks{Institute of Science and Engineering, Kanazawa University, Kakuma-machi, Kanazawa 920-1192, Japan. hmoriya4@se.kanazawa-u.ac.jp}
~Yu  Nakayama\thanks{Department of Physics, Rikkyo University, Toshima, Tokyo 171-8501, Japan}}
\usepackage{latexsym}
\usepackage{amssymb}
\usepackage{amsmath} %
\usepackage{amsthm} %
\usepackage{amscd} %
\usepackage{mathrsfs} 
\newtheorem{thm}{Theorem}[section]

\newtheorem{lem}[thm]{Lemma}
\newtheorem{prop}[thm]{Proposition}

\theoremstyle{definition}
\newtheorem{defn}[thm]{Definition}
\theoremstyle{remark}
\newtheorem{rem}[thm]{{\it{Remark}}}
\numberwithin{equation}{section}
\newcommand{\Z}{{\mathbb{Z}}}%
\newcommand{\NN}{{\mathbb{N}}}%
\newcommand{\id}{{\mathbf{1} }}
\newcommand{\unit}{1}
\newcommand{\Al}{\mathcal{A}}%
\newcommand{\core}{\Al_{\circ}}%
\newcommand{\Ale}{\Al_{+} }%
\newcommand{\Alo}{\Al_{-}}%

\newcommand{\AlI}{{\Al}({\I})}%
\newcommand{\AlIe}{\AlI_{+}}%
\newcommand{\AlIo}{\AlI_{-}}%
\newcommand{\Ali}{{\Al}(\{i\})}%
\newcommand{\coree}{{\core}_+}%
\newcommand{\coreo}{{\core}_-}%
\newcommand{\cicr}{c_i^{\,\ast}}%
\newcommand{\ci}{c_i}%
\newcommand{\cjcr}{c_j^{\,\ast}}%
\newcommand{\cj}{c_j}%
\newcommand{\ctwokcr}{c_{2k}^{\, \ast}}%
\newcommand{\ctwok}{c_{2k}}%
\newcommand{\ctwokpcr}{c_{2k+1}^{\, \ast}}%
\newcommand{\ctwokp}{c_{2k+1}}%
\newcommand{\ctwolcr}{c_{2l}^{\, \ast}}%
\newcommand{\ctwol}{c_{2l}}%
\newcommand{\ctwolmcr}{c_{2l-1}^{\, \ast}}%
\newcommand{\ctwolm}{c_{2l-1}}%
\newcommand{\del}{\delta}
\newcommand{\delQ}{\del_Q}
\newcommand{\ket}[1]{|#1 \rangle}
\newcommand{\I}{{\mathrm{I}}}%
\newcommand{\zetai}{\zeta_{i}}
\newcommand{\opkappai}{\hat{\kappa}_{i}}
\newcommand{\opkappaiminus}{\hat{\kappa}_{i-1}}
\newcommand{\opkappaiplus}{\hat{\kappa}_{i+1}}
\newcommand{\Ikl}{{\rm{I}}_{k,l}}
\newcommand{\Izeroone}{{\rm{I}}_{0,1}}
\newcommand{\Izerotwo}{{\rm{I}}_{0,2}}
\newcommand{\Izerothree}{{\rm{I}}_{0,3}}
\newcommand{\Izeron}{{\rm{I}}_{0,n}}
\newcommand{\OPERA}{\hat{\mathcal{O}}}%
\newcommand{\Qseq}{\mathscr{Q}}%
\newcommand{\Xiconpl}{\hat{\Xi}_{\,+1\rm{const.}}}
\newcommand{\Xiconmi}{\hat{\Xi}_{\,-1\rm{const.}}}
\newcommand{\Xicon}{\hat{\Xi}_{\rm{const.}}}
\newcommand{\IDplkl}{r^{+}_{[2k,2l]}}
\newcommand{\IDmikl}{r^{-}_{[2k,2l]}}
\newcommand{\IDplzeronone}{r^{+}_{[0,2(n+1)]}}
\newcommand{\IDplzeroone}{r^{+}_{[0,2]}}
\newcommand{\IDmizeroone}{r^{-}_{[0,2]}}
\newcommand{\IDplnnone}{r^{+}_{[2n,2(n+1)]}}
\newcommand{\IDplzerotwo}{r^{+}_{[0,4]}}
\newcommand{\IDmizerotwo}{r^{-}_{[0,4]}}
\newcommand{\IDplzerothree}{r^{+}_{[0,6]}}
\newcommand{\IDmizerothree}{r^{-}_{[0,6]}}
\newcommand{\IDplonetwo}{r^{+}_{[2,4]}}
\newcommand{\IDmionetwo}{r^{-}_{[2,4]}}
\newcommand{\IDpltwothree}{r^{+}_{[4,6]}}
\newcommand{\IDmitwothree}{r^{-}_{[4,6]}}
\newcommand{\IDplonethree}{r^{+}_{[2,6]}}
\newcommand{\IDmionethree}{r^{-}_{[2,6]}}
\newcommand{\cstar}{{{C}}^{\ast}}%
\newcommand{\Ome}{\Omega}%
\begin{document}
\bibliographystyle{unsrtnat}
\maketitle

\noindent{\bf{keywords}}: 
{Supersymmetric  fermion lattice model.  Ground states. 
Symmetry breakdown.}
\begin{abstract}
We study  a  supersymmetric fermion lattice model   
 defined  by  Hermann Nicolai.
We show that  its   infinitely many  classical supersymmetric ground states 
 are  associated to   breakdown of    
 hidden  local supersymmetries.
 \end{abstract}

\section{Introduction}
\label{sec:INTRO}
 A supersymmetric  fermion lattice model  defined  by Nicolai \cite{NIC} 
 is  a   pioneering  work 
 on (non-relativistic) supersymmetric quantum mechanics.
  This model, which we  call  Nicolai model,  even predates Witten's 
 supersymmetric  quantum mechanical model   \cite{WITT81}, 
 see \cite{JUNK} \cite{JUNK40} for some historical remarks. 
 It has been  shown that the Nicolai model has  highly degenerate 
 supersymmetric ground states \cite{MORIYA-JS} \cite{RUBEN} which give rise 
 to interesting dynamical properties.
The aim of this paper is to discuss  the    degeneracy 
 of supersymmetric  ground states  of the Nicolai model
 from the viewpoint of  symmetry breakdown. 
 Based on  our previous findings \cite{MORIYA-JS}
 we will classify all classical  supersymmetric  ground states   
 in terms of  breakdown of  local fermionic symmetries (supersymmetries) 
 hidden in the  model.

\subsection{Supersymmetric fermion lattice model by Nicolai}
\label{subsec:MODEL}
We  introduce   a  spinless fermion  lattice model on  one-dimensional integer lattice $\Z$   given  by Nicolai \cite{NIC}. 
 For each  site $i\in \Z$
let $\ci$ and $\cicr$ denote the annihilation  
and the creation  of a spinless fermion at $i$. 
 They obey  the canonical anticommutation relations:
For  all $i,j\in\Z$
\begin{align}
\label{eq:CAR}
\{ \cicr, \cj \}&=\delta_{i,j}\, \unit, \nonumber \\
\{ \cicr, \cjcr \}&=\{ \ci, \cj \}=0.
\end{align}
For each site  $i\in \Z$  the fermion number operator is defined by
\begin{align}
\label{eq:ni}
n_i:=\cicr \ci.
\end{align}
A formal  infinite sum  $N:=\sum_{i\in\Z}n_i$ will denote 
 the total fermion number operator.
Let      
\begin{align}
\label{eq:Q-NIC}
Q=\sum_{i\in \Z} q_{2i},\ \ \text{where}\ q_{2i}
:= c_{2i+1} c^{\ast}_{2i} c_{2i-1}.
\end{align}
 We see  that
   $Q$ and its adjoint $Q^{\ast}$ are  fermion operators in the sense that 
\begin{equation}
\label{eq:Q-NICodd}
\{(-1)^{N},\ Q\}= \{(-1)^{N},\ {Q}^{\ast}\}=0.
\end{equation}
It is essential that   the  nilpotent property is satisfied:     
\begin{align}
\label{eq:Q-NICnil}
{Q}^{2}=0={{Q}^{\ast}}^{2}.
\end{align}
The supersymmetric Hamiltonian is given by
\begin{align}
\label{eq:HNIC}
H&:=\{ Q,\; {Q}^{\ast} \}.
\end{align}
The pair of  supercharges
$Q$, ${Q}^{\ast}$ and  
 the  supersymmetric  Hamiltonian 
 $H$  satisfy  the ${\cal{N}}=2$   supersymmetry relation 
   \cite{WEIN}, although there is no boson  in the model.

The  explicit  form of $H$ can be easily computed as 
\begin{align}
\label{eq:HNIC-gutai}
H
&=\sum_{i\in \Z}\bigl\{ 
c^{\ast}_{2i}c_{2i-1}c_{2i+2}c^{\ast}_{2i+3}
+c^{\ast}_{2i-1}c_{2i}c_{2i+3}c^{\ast}_{2i+2} \nonumber\\
&\ +c^{\ast}_{2i}c_{2i}c_{2i+1}c^{\ast}_{2i+1}
+c^{\ast}_{2i-1}c_{2i-1}c_{2i}c^{\ast}_{2i}
- c^{\ast}_{2i-1} c_{2i-1}c_{2i+1}c^{\ast}_{2i+1} \bigr\}.
\end{align}
 The Nicolai model has  some   obvious   symmetries. 
The global $U(1)$-symmetry group   $\gamma_{\theta}$ 
 $(\theta \in [0,\; 2\pi])$ is defined  by    
\begin{align}
\label{eq:U1}
\gamma_{\theta} (c_{i})=e^{-i\theta}c_{i},\quad 
\gamma_{\theta} (c_{i}^{\ast})=e^{i\theta}c_{i}^{\ast},
\quad \forall i  \in \Z.
\end{align}
The  particle-hole transformation   is given by the  
 $\Z_{2}$ action: 
\begin{align}
\label{eq:particle-hole}
\rho (c_{i})=c_{i}^{\ast},\quad 
\rho (c_{i}^{\ast})=c_{i},
\quad \forall i  \in \Z.
\end{align}
 Let $\sigma$ denote the shift-translation automorphism 
 group  defined by 
\begin{align}
\label{eq:sigk}
\sigma_{k} (c_{i})=c_{i+k},\quad 
\sigma_{k} (c_{i}^{\ast})=c_{i+k}^{\ast}
\quad \forall i  \in \Z,\  \text{for each}\ k\in\Z.
\end{align}
The  Hamiltonian $H$
 \eqref{eq:HNIC-gutai}
  is invariant under 
$\gamma_\theta$  $(\theta \in [0,\; 2\pi])$,
 so it has   the global $U(1)$-symmetry.  
It has particle-hole symmetry as  $\rho(H) = H$, which follows from $\rho(Q)=-Q^\ast$ and $\rho(Q^\ast)=-Q$.
Finally, $H$ is invariant under translation by two sites, 
 as  $\sigma_{2k}(H)=H$ for any  $k\in \Z$, whereas 
  the  full  translation symmetry is explicitly  broken as    
 $\sigma_{2k+1}(H)\ne H$.
We will see 
 in $\S$\ref{sec:LOCAL} that  the Nicolai model has  other  
 local   symmetries.

\subsection{Mathematical preliminary}
In this subsection, we introduce some basic notations.
We refer to \cite{MORIYAahp} that gives 
 a general framework of  supersymmetric fermion lattice systems.
Although  it is not  absolutely necessary,   
 the  $\cstar$-algebraic  formulation    is 
 helpful to  formulate   our pertinent problem
 and gives a clue to solve it.

For each finite $\I\Subset\Z$, $\AlI$ denotes  
  the  finite-dimensional algebra generated by 
$\{\cicr, \, \ci\, ;\;i\in \I\}$, where the notation `$\I \Subset \Z$' means that  $\I\subset\Z$ and the number of sites $|\I|$ in $\I$ is finite.
The union of all these  $\AlI$  defines  the local algebra: 
\begin{equation} 
\label{eq:CARloc}
\core:=\bigcup_{\I \Subset \Z }\AlI.
\end{equation}
The norm completion of the local algebra $\core$ 
 gives a $\cstar$-system  $\Al$  called  the CAR algebra.    

Let $\Theta$ denote  the fermion grading automorphism on $\Al$ given as:
  \begin{equation}
\label{eq:CARTHETA}
\Theta(\ci)=-\ci, \quad \Theta(\cicr)=-\cicr,\quad \forall i\in \Z.
\end{equation}
 The  fermion system $\Al$ 
 is decomposed into the even part $\Ale$ and the odd part $\Alo$ as
\begin{align}
\label{eq:grad}
\Al&=\Ale\oplus \Alo,\quad 
\Ale= \{A\in \Al| \; \Theta(A)=A\},\quad
\Alo= \{A\in \Al|\;  \Theta(A)=-A\}.
\end{align}
Any element of  $\Ale$ is a linear sum of 
even  monomials of fermion field operators, while 
  that of  $\Alo$ is a linear sum of 
odd monomials of fermion field operators.
Similarly, for each  $\I\Subset\Z$, 
\begin{equation}
\label{eq:CARIeo}
\AlI=\AlIe\oplus\AlIo,\ \  
 \AlIe := \AlI\cap \Ale,\quad  \AlIo := \AlI\cap \Alo,
 \end{equation}
and  for the local algebra
\begin{equation}
\label{eq:core-grad}
\core=\coree\oplus\coreo,\ \ 
 \coree := \core \cap \Ale,\quad  \coreo := \core\cap \Alo.
 \end{equation}

Define  the  graded commutator 
$[\; \ , \; \ ]_{\Theta}$ 
  on $\Al$ by the  mixture of the commutator $[\ ,\ ]$
 and the anti-commutator $\{\ , \ \}$ as
\begin{align}
\label{eq:gcom}
[A_{+}, \;  B]_{\Theta}&= [A_{+}, \;  B]\equiv 
A_{+}B-BA_{+}
 {\text {\ \ for \ }}
A_{+} \in \Ale,    \ B \in \Al, \notag\\
[A, \;  B_{+}]_{\Theta} &= [A, \;  B_{+}]
=
AB_{+}-B_{+}A {\text {\ \ for \ }}
A \in \Al,   \ B_{+} \in \Ale, \notag\\
[A_{-}, \;  B_{-}]_{\Theta} &= \{A_{-},  \; B_{-}\}\equiv A_{-}B_{-}
+B_{-}A_{-}
  {\text {\ \ for \ }} A_{-}, B_{-}\in \Alo. 
\end{align}

Consider  the superderivation
 generated by the nilpotent supercharge $Q$:  
\begin{equation}
\label{eq:del-intro}
\delQ(A):=[Q,\; A]_{\Theta}
{\text {\ \ for   \ }} A \in \core. 
\end{equation}
We  see   that $\delQ$ 
   is a linear map that anticommutes with the grading:  
\begin{equation}
\label{eq:oddsuperder}
\delQ\cdot \Theta=-\Theta \cdot \delQ, 
\end{equation}
and that the {\it{graded Leibniz rule}} holds:
\begin{equation}
\label{eq:gleib}
\delQ(AB)=
\delQ(A)B+\Theta(A) \delQ(B)\ \  {\text {for }}\ A, B\in \Al.
\end{equation}

A state  (i.e. normalized positive linear functional on $\Al$) is called 
 a supersymmetric (ground) state
if and only if it is invariant under the superderivation $\delQ$, 
 equivalently  its state  vector (determined by  the GNS representation) 
is annihilated by both the supercharge $Q$
 and its adjoint $Q^\ast$. 
In this paper, we deal with  only pure states on finite systems 
 that are  always associated with 
   normalized vectors. 
 For a supersymmetric model, if there is a 
   supersymmetric  state, 
then the supersymmetry is unbroken. If there exists  no supersymmetric 
 state, then the supersymmetry  is spontaneously broken.  
 The Nicolai model has many supersymmetric  states as we will see later. 
 Hence  its supersymmetry  is unbroken.

\subsection{Classical supersymmetric ground  states of the Nicolai model}
\label{sec:CLASSICAL}

In this paper,  we focus on  
 {\it{classical}} supersymmetric ground states which will be stated below.
This subsection is indebted to \cite{MORIYA-JS}.

Let 
$\ket{1}_{i}$ and 
$\ket{0}_{i}$ denote the  occupied and empty  
 vectors  of the spinless fermion  at site $i$, respectively.
For each $i\in\Z$
\begin{equation}
\label{eq:CARsimple}
\ci \ket{1}_{i}= \ket{0}_{i},\ 
\cicr \ket{1}_{i}=0,\ 
\cicr \ket{0}_{i}= \ket{1}_{i},\ 
\ci \ket{0}_{i}=0.
\end{equation}
When there is no fear of confusion,  we will omit the  subscript 
and write  simply  
$\ket{1}$ and $\ket{0}$.

We identify    general (not necessarily supersymmetric) 
classical states on the fermion lattice system
  by  classical configurations on $\Z$.
\begin{defn}
\label{defn:CLASSIC-config}
Let $g(n)$ denote an arbitrary   
$\{0, 1\}$-valued  function  over  $\Z$. 
It is called  a classical configuration over $\Z$.
For any classical configuration $g(n)$ define 
\begin{align}
\label{eq:gn-vector}
\ket{g(n)_{n\in\Z}}:= \cdots \otimes 
\ket{g(i-1)}_{i-1} \otimes  \ket{g(i)}_{i}
\otimes \ket{g(i+1)}_{i+1} \otimes   \cdots.
\end{align}
This infinite product vector determines  
a  state $\psi_{g(n)}$ on the  fermion system  $\Al$ 
 which  will be  called  the  classical state associated 
to the configuration  $g(n)$ over $\Z$.
Let   $\iota_0(n):=0$  $\forall n\in \Z$.
Then
\begin{align}
\label{eq:Fock-via-config}
\Omega_{0}:=
\ket{\iota_0(n)_{n\in\Z}}=
 \cdots \otimes \ket{0}
\otimes \ket{0}
\otimes \ket{0} \otimes  \ket{0}
\otimes \ket{0} \otimes  \ket{0} \cdots.
\end{align}
 The above   $\Omega_{0}$ is called the Fock vector, 
 and its associated translation-invariant 
state $\psi_0$ on $\Al$ is called the Fock state.
Similarly let    $\iota_1(n):=1$  $\forall n\in \Z$.
Then
\begin{align}
\label{eq:OCCUP-via-config}
\Omega_{1}:=
\ket{\iota_1(n)_{n\in\Z}}=
 \cdots \otimes \ket{1}
\otimes \ket{1}
\otimes \ket{1} \otimes  \ket{1}
\otimes \ket{1} \otimes  \ket{1} \cdots.
\end{align}
 The above   $\Omega_{1}$ is called the 
fully-occupied vector, 
 and its associated translation-invariant 
state $\psi_1$ on $\Al$ is called the fully-occupied  state.
\end{defn}

 To each  classical configuration over $\Z$
 we   assign  an operator by the following rule. 
\begin{defn}
\label{defn:OPERA}
For each $i\in\Z$
  let $\opkappai$ denote the map 
from $\{0, 1\}$ into  $\Ali$ given as  
\begin{align}
\label{eq:OPERAi-config}
\opkappai(0):= \id, \quad \opkappai(1):= \cicr.
\end{align}
For each  classical configuration  $g(n)$ over $\Z$ 
 define   
the   infinite-product of fermion field operators: 
\begin{align}
\label{eq:defOPERAg}
\OPERA(g)&:=\prod_{i\in\Z} \opkappai\left(g(i)\right)=
\cdots \opkappaiminus\left(g(i-1)\right) \opkappai\left(g(i)\right) 
\opkappaiplus\left(g(i+1)\right)
\cdots, 
\end{align}
where the  multiplication  is  taken  in the  increasing  order of $i\in \Z$. 
If  $g(n)$ has  a  compact support, then   
\begin{align}
\label{eq:OPERAgcompact}
\OPERA(g)\in \core.
\end{align}
Otherwise  $\OPERA(g)$
 denotes  a  formal operator which does not belong to  $\Al$. 
\end{defn}

We have the following obvious correspondence between 
 product vectors given in 
Definition \ref{defn:CLASSIC-config} and product operators 
 given in  Definition \ref{defn:OPERA}  
via   the Fock representation.
\begin{prop}
\label{prop:Classic-vec-config}
Let  $\Omega_0$ denote the  Fock  vector 
 given in \eqref{eq:Fock-via-config}.
For any classical configuration  $g(n)$ over $\Z$, 
 the following  identity   holds{\rm{:}} 
\begin{align}
\label{eq:gonFock}
\OPERA(g)\Omega_0=\ket{g(n)_{n\in\Z}}.
\end{align}
\end{prop}
\begin{proof}
The desired identity   directly  follows 
from  Definition \ref{defn:CLASSIC-config}
 and Definition \ref{defn:OPERA}.
\end{proof}

It is easy to see that 
 the Fock state $\psi_0$ and the fully-occupied state $\psi_1$ are 
  supersymmetric  ground state for the Nicolai model. 
We would like to  give   all  classical 
supersymmetric ground states of the Nicolai model. 
For this purpose, 
 we introduce  the following  class of classical configurations.
 \begin{defn}
\label{defn:ground-config}
 Consider   three consecutive sites  $\{2i-1, 2i, 2i+1\}$ centered 
 at an even site $2i$ ($i\in\Z$).
   There are  $2^{3}$ configurations (i.e. eight $\{0,1\}$-valued functions)
on  $\{2i-1, 2i, 2i+1\}$. Let   
   $``0, 1, 0"$ and  $``1, 0, 1"$  be  called    forbidden triplets.
If  a classical configuration $g(n)$ $(n\in\Z)$ 
 does not include any of such  forbidden triplets   over $\Z$, 
then it is called a ground-state  configuration  over $\Z$ 
(for  the Nicolai model).
The set of all ground-state  configurations  over $\Z$
 is  denoted by  $\Upsilon$.
The set of all ground-state  configurations 
 whose supports are  included in some finite region 
 is denoted by  $\Upsilon_{\circ}$.
The set of all ground-state  configurations 
 whose supports are  included in 
a finite region $\I\Subset\Z$ is denoted by   $\Upsilon_{\I}$.
\end{defn}

  The following proposition    
 classifies  all the classical supersymmetric ground states in terms of 
classical configurations justifying 
   our   nomenclature ``ground-state configurations''  of  
   Definition \ref{defn:ground-config}.
It is based on the following  fact 
that can be easily checked by using  \eqref{eq:CARsimple}:
 The product vector 
$\ket{g(2i-1)}_{2i-1} \otimes  \ket{g(2i)}_{2i}
\otimes \ket{g(2i+1)}_{2i+1} 
$
is annihilated by both $q_{2i}$ and $q_{2i}^{\ast}$
 unless those are 
$
\ket{0}_{2i-1} \otimes  \ket{1}_{2i}
\otimes \ket{0}_{2i+1} 
$
or 
$
\ket{1}_{2i-1} \otimes  \ket{0}_{2i}
\otimes \ket{1}_{2i+1} 
$
which correspond to 
the forbidden triplets,  
 $\{g(2i-1)=0,\; g(2i)=1,\; g(2i+1)=0\}$ and   
  $\{g(2i-1)=1,\; g(2i)=0,\; g(2i+1)=1\}$, respectively.  
Theorem 2 \cite{MORIYA-JS} established that 
if there appears  no forbidden triplet in the sequence of  $g(n)$ $(n\in\Z)$  at all, then the corresponding product vector  
$\OPERA(g)\Omega_0$ is annihilated by both 
 $Q$ and  $Q^\ast$, hence it is a supersymmetric ground state, whereas  
if there is at least one   forbidden triplet in the sequence of  $g(n)$ $(n\in\Z)$, then either $Q$
 or  $Q^\ast$, or both do not annihilate 
$\OPERA(g)\Omega_0$ and so it is not   supersymmetric.
See \cite{MORIYA-JS} for the detail.
\begin{prop}
\label{prop:AllclassicZ}
A classical state  on the fermion lattice system  $\Al$
  is supersymmetric for the Nicolai model if and only if 
its associated configuration  
 $g(n)$ over $\Z$ is a ground-state 
 configuration for the Nicolai model as stated in  Definition
   \ref{defn:ground-config}, namely, if and only if 
   $g(n)\in \Upsilon$.
\end{prop}

\subsection{Supersymmetric ground states on  subsystems}
\label{subsec:SUSYonI}
We shall discuss   supersymmetric ground states on finite subsystems.
 First,  we  specify finite regions that we will consider.
Second, we  specify the  meaning of 
  ``supersymmetric ground states"  upon finite regions, as it  
is  not so obvious due to the boundary.

To deal with the Nicolai model which has period-2 translational symmetry 
 not    full  translation symmetry it is convenient to  
 consider  the special finite  intervals of $\Z$ 
whose  edges are both even, see
Proposition  \ref{prop:Iklenough} given later.
 Namely for  $k,l\in\Z$ ($k<l$) we  take 
\begin{align} 
\label{eq:Ikl}
\Ikl\equiv[2k,\; 2k+1,\;2(k+1),  \cdots\cdots,2(l-1),\; 2l-1, \;,2l].
\end{align}
We see that  $|\Ikl|$  is $2(l-k)+1$. 

Now we give the precise definition of supersymmetric ground states 
 on the finite interval $\Ikl$.
 \begin{defn}
\label{defn:SUSYstatesIkl}
 Consider  any finite interval 
 $\Ikl \equiv[2k,\; 2k+1,\;2(k+1),  \cdots,2(l-1),\; 2l-1, \;2l]$ $(k,l\in\Z$ 
   $k<l)$. 
Let      
\begin{align}
\label{eq:Q-NICklopen}
Q[k, l]\equiv \sum_{i=k}^{l} q_{2i}\in  {\Al}\left(\{2k-1\}\cup\Ikl\cup
\{2l+1\}\right)_{-},
\end{align}
where  $q_{2i}\equiv -c_{2i-1} c^{\ast}_{2i} c_{2i+1}$ as defined 
 in \eqref{eq:Q-NIC}. 
A state on $\Al(\Ikl)$ is called a 
{{\it{free-boundary supersymmetric ground state}}}
 if its arbitrary  state-extension to ${\Al}\left(\{2k-1\}\cup\Ikl\cup
\{2l+1\}\right)$ is 
invariant under the superderivation $\del_{Q[k, l]}$
 associated to the local supercharge $Q[k, l]$, equivalently its 
associated vector is 
 annihilated by both  $Q[k, l]$ and  $Q[k, l]^{\ast}$. 
\end{defn}

 First note that 
\begin{align}
\label{eq:Qidentical}
\del_{{Q[k, l]}}= \del_{Q}\  \text{on}\  \Al(\Ikl),
\end{align}
where $Q=\sum_{i\in \Z} q_{2i}$
 as  in \eqref{eq:Q-NIC}.
Namely,  upon  the  subsystem $\Al(\Ikl)$,  finite supercharge  ${Q[k, l]}$ 
 sitting on a slightly larger region $\{2k-1\}\cup\Ikl\cup
\{2l+1\}$  gives  the same  action as of  the total supercharge $Q$.
Second, we  address what ``its arbitrary  state-extension'' exactly means.
According to  \cite{AM2003}, 
for  every  classical state on the given local system, 
any  state-extension of it to  a larger system  
 is  also a   classical state (or  mixture of such).
In the present case  it is described as follows. 
By Proposition \ref{prop:Classic-vec-config}
 any classical state of 
$\Al(\Ikl)$ is determined by 
a  $\{0, 1\}$-valued  function $g(n)$  on  $\Ikl$. 
Any  state-extension of it    
to ${\Al}\left(\{2k-1\}\cup\Ikl\cup
\{2l+1\}\right)$ is determined by 
$\tilde{g}(n)$  on $\{2k-1\}\cup\Ikl\cup\{2l+1\}$
 satisfying that 
\begin{align}
\label{eq:extension}
\tilde{g}(n)=g(n)\ \text{for}\ \forall n \in \Ikl,\ \ 
\tilde{g}(2k-1)=0\ \text{or}\ 1,\   
\tilde{g}(2l+1)=0\ \text{or}\ 1.
\end{align}  
Due to the choice of the marginal points
 $\{2k-1, 2l+1\}$ there are  four possibilities.

To find  configurations associated to  Definition
   \ref{defn:SUSYstatesIkl}, 
 we  introduce  a subclass  of  
$\Upsilon_{\circ}$ given in  Definition \ref{defn:ground-config}
 requiring  certain   boundary  conditions as follows.
\begin{defn}
\label{defn:SUSYboundaryconfig}
Let  $\Ikl \equiv[2k,\; 2k+1,\;2(k+1),  \cdots,2(l-1),\; 2l-1, \;2l]$  
($k,l\in\Z$ s.t. $k<l$) as before. 
The set of all 
  $g(n)\in \Upsilon_{k,l}\equiv{\Upsilon_{\Ikl}}$   satisfying 
the following boundary conditions  
\begin{equation}
\label{eq:edge-const}
g(2k)=g(2k+1) \ \text{and} \  
g(2l-1)=g(2l)
\end{equation}
  will be  denoted by  $\widehat{\Upsilon}_{k,l}$.
\end{defn}

In Proposition  \ref{prop:AllclassicZ}
we gave     one-to-one correspondence  between
 classical supersymmetric ground states 
 and ground-state configurations over $\Z$.  
We can see   analogous correspondence 
 on  finite regions $\Ikl$  as  follows.

\begin{prop}
\label{prop:classical-Ikl}
A  classical state  on the finite system  $\Al(\Ikl)$
  is free-boundary supersymmetric (Definition
   \ref{defn:SUSYstatesIkl}) if and only if  its   associated configuration
 $g(n)$ on  $\Ikl$  belongs to $\widehat{\Upsilon}_{k,l}$ 
(Definition \ref{defn:SUSYboundaryconfig}).
\end{prop}

\begin{proof}
We will  see the if part as follows.
For all $i\in\{k+1, k+2, \cdots, l-1\}$ 
both $q_{2i}$ and $q_{2i}^{\ast}$  annihilate 
any product vector corresponding to 
$g(n)\in \Upsilon_{k,l}$.
 So we only have to see the marginal points $k$ and $l$.  
For any  given $g(n)\in \widehat{\Upsilon}_{k,l}$ its extension 
 to $\{2k-1\}\cup\Ikl\cup \{2l+1\}$  will be denoted as 
 $\tilde{g}(n)$. We see that 
$\tilde{g}(2k-1)$ is  arbitary, 
$\tilde{g}(2k)=\tilde{g}(2k+1)$, 
$\tilde{g}(2l-1)=\tilde{g}(2l)$, and  
 $\tilde{g}(2l+1)$ is  arbitrary.
So there is no forbidden sequence on 
$\{2k-1, 2k, 2k+1\}$. Thus  
 both $q_{2k}$ and $q_{2k}^{\ast}$  annihilate 
any product vector corresponding to 
$\tilde{g}(n)$.
Similarly, both $q_{2l}$ and $q_{2l}^{\ast}$  annihilate 
 the  product vector corresponding to $\tilde{g}(n)$.
The only if part can be shown as in  Theorem 2 \cite{MORIYA-JS}.
\end{proof}

\section{Hidden local fermionic symmetries}
\label{sec:LOCAL}
We will show that there are  infinitely many local fermionic symmetries  
 hidden in  the  Nicolai model.
To this end,  we need some  preparation. 
\begin{defn}
\label{defn:seq-conservation}
Take any  finite interval  $\Ikl$   defined in \eqref{eq:Ikl}.
Let $f$ be a $\{-1, +1\}$-valued  sequence  on   $\Ikl$. 
For any consecutive triplet    
$\{2i-1,\; 2i,\; 2i+1\}\subset \Ikl$ ($i\in\Z$)
assume that  neither
\begin{align}
\label{eq:2i-forbid}
f(2i-1)=-1,\ \ f(2i)=+1,\ \  f(2i+1)=-1  
\end{align}
nor 
\begin{align}
\label{eq:2i-forbid-ura}
f(2i-1)=+1,\ \ f(2i)=-1,\ \ f(2i+1)=+1
\end{align}
holds. Furthermore assume  that $f$ is constant  on the 
left-end  pair sites $\{2k,\; 2k+1\}$ and 
on the right-end  pair sites $\{2l-1,\; 2l\}$:
\begin{align}
\label{eq:leftedge-seq}
f(2k)=f(2k+1)=+1\quad  \text{or}\quad 
f(2k)=f(2k+1)=-1
\end{align}
and 
\begin{align}
\label{eq:rightedge-seq}
f(2l-1)=f(2l)=+1\quad  \text{or}\quad 
f(2l-1)=f(2l)=-1.
\end{align}
The set of all  
$\{-1, +1\}$-valued sequences
on $\Ikl$ satisfying  the above  conditions 
 is  denoted by  $\hat{\Xi}_{k,l}$.
The union of   $\hat{\Xi}_{k,l}$ over
all   $k,l\in\Z$ ($k<l$) is  denoted by $\hat{\Xi}$:
\begin{align}
\label{eq:hatkl-inner}
\hat{\Xi}
:=\bigcup_{k,l\in \Z \;  (k<l)} \hat{\Xi}_{k,l}.
\end{align}
Take  any $p,q\in\Z$  $(p<q)$.
 Let 
\begin{align}
\label{eq:hatkl-inner}
\hat{\Xi}(p,q):=\bigcup_{k, l\in \Z;\;
p\le k 
< l \le  q } \hat{\Xi}_{k,\; l}.
\end{align}
Each $f\in\hat{\Xi}$ 
  is called a   local  $\{-1, +1\}$-sequence of conservation for 
the Nicolai model.
\end{defn}

\begin{rem} 
\label{rem:marginal}
The  requirements 
\eqref{eq:leftedge-seq}
\eqref{eq:rightedge-seq} on the edges of $\Ikl$  are   essential to 
 make  conservation laws   for the Nicolai model.
 \end{rem}

\begin{rem} 
\label{rem:crude}
By  crude estimate 
we  can see that  
the number of local   $\{-1, +1\}$-sequences of  conservation
  in  
 $\hat{\Xi}_{k,l}$ is roughly $(\frac{2^{3}-2}{2})^{(l-k)}=3^{(l-k)}=3^{m/2}$, 
 where  $m=2(l-k)$  denotes approximately the  size of the system 
(i.e.  the number of sites in  $\Ikl$).
 \end{rem}

It is convenient to  consider   the following   subclasses 
 of $\hat{\Xi}$. 
\begin{defn}
\label{defn:ID-conserved}
For each $k,l\in\Z$ ($k<l$) let $\IDplkl\in \hat{\Xi}_{k,l}$ 
 and $\IDmikl\in \hat{\Xi}_{k,l}$ denote  the constants over $\Ikl$ as 
\begin{equation}
\label{eq:constkl}
\IDplkl(i)=+1 \  \forall i\in \Ikl,\quad    
\IDmikl(i)=-1 \  \forall i\in \Ikl. 
\end{equation}
The set    
$\left\{\IDplkl\right\}$
over all   $k,l\in\Z$ ($k<l$) will be denoted as $\Xiconpl$,
 and  the set   $\left\{ \IDmikl \right\}$
over all   $k,l\in\Z$ ($k<l$) will be denoted as $\Xiconmi$.
 Let $\Xicon:=\Xiconpl\cup \Xiconmi$.
\end{defn}

\vspace{5mm}

We shall  give  a rule to assign a 
  local fermion operator  for every    
local $\{-1, +1\}$-sequence   of conservation 
 of   Definition \ref{defn:seq-conservation}.
\begin{defn}
\label{defn:local-assignment}
For each $i\in\Z$
  let $\zetai$ denote the assignment   
from $\{-1,+1\}$ into  the fermion field at $i$ given as  
\begin{align}
\label{eq:zeta}
\zetai(-1):= \ci,\quad \zetai(+1):= \cicr.
\end{align}
 Take  any  pair of integers $k,l\in\Z$ such that $k<l$.
For  each  $f\in \hat{\Xi}_{k,l}$, set 
\begin{align}
\label{eq:Qseq}
\Qseq(f)
&:=\prod_{i=2k}^{2l} \zetai\left(f(i)\right)\nonumber\\
&\equiv \zeta_{2k} 
\left(f(2k)\right) \zeta_{2k+1}\left(f(2k+1)\right) 
\cdots 
\cdots \zeta_{2l-1}\left(f(2l-1)\right) 
\zeta_{2l}\left(f(2l)\right)\in 
{\Al}(\Ikl)_{-},
\end{align}
 where the  multiplication 
  is  taken in the  increasing  order as above. 
The  formulas \eqref{eq:Qseq} for  all   $k,l\in\Z$ ($k<l$)
   yield   a unique   assignment $\Qseq$  from $\hat{\Xi}$ to $\coreo$.  
\end{defn}

By   Definition \ref{defn:local-assignment}
 the following  local fermion operators
are assigned to   $\pm$-characters supported on the segment $\Ikl$
   of   Definition   \ref{defn:ID-conserved}. 
  For $k,l\in\Z$ ($k<l$)
\begin{align}
\label{eq:Qseqconst}
\Qseq(\IDplkl)&:=\ctwokcr \ctwokpcr \cdots 
\ctwolmcr \ctwolcr \in 
{\Al}(\Ikl)_{-}, \nonumber\\
\Qseq(\IDmikl)&:=\ctwok \ctwokp \cdots 
\ctwolm \ctwol \in 
{\Al}(\Ikl)_{-}.
\end{align}

The following is the main result of this section.
\begin{thm}
\label{thm:NICconstant}
For  every $f\in\hat{\Xi}$ 
\begin{equation}
\label{eq:naive-fconst}
[H,\; \Qseq(f)]=0=
[H,\; \Qseq(f)^{\ast}],
\end{equation}
where $H$ denotes the Hamiltonian of the Nicolai model 
 over $\Z$.
\end{thm}

\begin{proof}
This theorem is established  in  \cite{MORIYA-JS}.
Because of its   importance and the reader's convenience,  we will provide 
 its  more formal derivation  below.  
It suffices to show that 
 \begin{equation}
\label{eq:QantiQloc}
\{Q,\ \Qseq(f)\}= \{{Q}^{\ast},\ \Qseq(f)\}=0,
 \end{equation}
and that 
 \begin{equation}
\label{eq:QantiQlocast}
\{Q,\ \Qseq(f)^{\ast}\}= \{{Q}^{\ast},\ \Qseq(f)^{\ast}\}=0,
 \end{equation}
 as the former  implies $[H,\; \Qseq(f)]=0$
 and the latter implies $[H,\; \Qseq(f)^{\ast}]=0$  by 
the graded Leibniz rule of superderivations \eqref{eq:gleib}.
Recall  $Q=\sum_{i\in \Z} q_{2i}$  and 
$q_{2i}\equiv c_{2i+1} c^{\ast}_{2i} c_{2i-1}$ defined in \eqref{eq:Q-NIC}.
 By Definitions \ref{defn:seq-conservation} \ref{defn:local-assignment}, 
  we have 
\begin{align}
\label{eq:vanishQseq}
\Qseq(f)q_{2i}=0=q_{2i}\Qseq(f),\quad 
 \Qseq(f)q_{2i}^{\ast}=0=q_{2i}^{\ast}\Qseq(f)
 \quad  \text{for all}\quad i\in\Z.
\end{align}
From the above  we obtain 
\eqref{eq:QantiQloc} and \eqref{eq:QantiQlocast}.
\end{proof}

Theorem \ref{thm:NICconstant}  
 says   that  the  Nicolai  
model  has infinitely many  local fermionic constants.
Those generate local fermionic  symmetries. 
\begin{defn}
\label{defn:const-motion}
For each local  $\{-1, +1\}$-sequence of conservation  
 $f\in\hat{\Xi}$,  
$\Qseq(f)$ is called the  local fermionic constant of 
   motion  associated  to $f$, and 
the pair 
 $\{\Qseq(f), \Qseq(f)^{\ast}\}$ 
 is called the local fermionic charge
 associated  to $f$.
\end{defn}

\begin{rem} 
\label{rem:kinematical}
Fermionic symmetry satisfying the supersymmetry relation   
other than the dynamical supersymmetry is sometimes called kinematical 
 supersymmetry. See e.g. \cite{NAKChern-Simons}.
Hence the   local fermionic charge  $\{\Qseq(f), \Qseq(f)^{\ast}\}$ 
for any  $f\in\hat{\Xi}$
 gives a local kinematical  supersymmetry.
 \end{rem}

\begin{rem} 
\label{rem:gauge}
 Any operator  of the algebra 
   generated by 
 $\{\Qseq(f)\in \core|\; f\in \hat{\Xi} \}$
 is  a local constant of motion. 
There exist   many such  self-adjoint bosonic  
 operators that  generate  (bosonic) symmetries for the 
 Nicolai model. 
 For example, 
  we obtain a bosonic constant 
$ n_{2k} n_{2k+1}\cdots 
n_{2l-1}n_{2l} \in 
{\Al}(\Ikl)_{+}$ from the multiple 
$\Qseq(\IDplkl)
\Qseq(\IDmikl)$ where each of them is 
 defined in \eqref{eq:Qseqconst}.
 \end{rem}

\section{Degenerate classical supersymmetric ground states and broken 
local fermionic symmetries}
\label{sec:generation}
In  this section we will  
 relate the high degeneracy of ground states  
shown in $\S$ \ref{sec:CLASSICAL}
 to the existence of many local fermionic  symmetries 
shown in $\S$ \ref{sec:LOCAL}.
In particular, we will show  that 
 every  classical supersymmetric ground state can be constructed from  
 (broken) local fermionic  symmetries.

\begin{thm}
\label{thm:gene-allclassics-Ikl}
Take any  segment  $\Ikl$ 
 indexed by $k,l\in\Z$ ($k<l$) as in  \eqref{eq:Ikl}.
Any classical free-boundary supersymmetric ground 
 state on $\Al(\Ikl)$  (Definition
   \ref{defn:SUSYstatesIkl}) can be constructed by  some finitely many 
 applications  of  operators  $\Qseq(f)$ (and $\Qseq(f)^{\ast}$)    
 with $f\in  \hat{\Xi}(k,l)$  (Definition \ref{defn:seq-conservation})
 to  the Fock vector $\Omega_0$ \eqref{eq:Fock-via-config}, 
similarly, to the fully-occupied vector  $\Omega_1$ 
 \eqref{eq:OCCUP-via-config}.
\end{thm}

By  Definition \ref{defn:CLASSIC-config}
 we can  identify  every  classical supersymmetric ground state
 $\psi_{g(n)}$ on $\Al$ with  its corresponding 
classical configuration $g(n)$  over  $\Z$, and vice versa.
By Proposition \ref{prop:classical-Ikl} we can identify 
 the set of all  classical free-boundary supersymmetric ground states
 on $\Al(\Ikl)$ (Definition \ref{defn:SUSYstatesIkl}) 
  with $\widehat{\Upsilon}_{k,l}$  (Definition \ref{defn:SUSYboundaryconfig}). 
 We will  frequently use those identifications in what follows.

The following lemma  implies  that the latter  part (using $\Omega_1$) of  
Theorem \ref{thm:gene-allclassics-Ikl} holds
once the former part (using $\Omega_0$) is proved. 
Also it will    be used  frequently in  the proof.
\begin{lem}
\label{lem:vacuum-occupied-start}
For   any  $n\in\NN$,    
if   $g\in \widehat{\Upsilon}_{0,n}$ 
can be constructed by some finitely many  applications of 
local fermionic  charges within $\Izeron$
 to the Fock vector $\Omega_0$.
Then it  can be constructed 
by some  applications of 
local fermionic  charges within $\Izeron$ to  
the fully-occupied state $\Omega_1$.
\end{lem}

\begin{proof}
For any $f\in \hat{\Xi}(0,n)$,  
 $-f\in \hat{\Xi}(0,n)$ by definition.  
From \eqref{eq:Qseq} in Definition 
 \ref{defn:local-assignment}
\begin{align}
\label{eq:minus-particle-hole}
\Qseq(-f)=\rho(\Qseq(f)), 
\end{align}
 where $\rho$ is the particle-hole
 transformation  defined in 
\eqref{eq:particle-hole}.
Thus by using the particle-hole transformation, 
we can use  $\Ome_0$ and $\Ome_1$ interchangeably. 
\end{proof}

Obviously it is enough to show Theorem
\ref{thm:gene-allclassics-Ikl}  by setting   $k=0$ and $l=\forall n\in \NN$ by 
 shift-translations. Thus we will prove the following.
\begin{prop}
\label{prop:induction}
For   any  $n\in\NN$,    
every  $g\in \widehat{\Upsilon}_{0,n}$ 
can be constructed by some  applications  of 
local fermionic charges within $\Izeron${\rm{:}}  
\begin{align}
\label{eq:set-of-actions}
\left\{\Qseq(f) ; \ 
f\in \hat{\Xi}(0,n)\equiv 
 \bigcup_{m \in\{1,2,\cdots,n\}}  \hat{\Xi}_{0,m} \right\}
\end{align}
 to  the Fock vector $\Omega_0$ \eqref{eq:Fock-via-config}.
\end{prop}

\begin{proof}
We need  concrete  forms of  elements in  $\hat{\Xi}$ which 
  are listed  in $\S$\ref{subsec:example-hatXi}.

First, let us consider the case $n=1$.
$\widehat{\Upsilon}_{0,1}$ consists of the following  two sequences
on $\Izeroone$:
\begin{center}
\begin{tabular}{|c||ccc|}
\hline
$\widehat{\Upsilon}_{0,1}$ &$0$ &$1$ &$2$  \\ \hline
$g^{\circ}_{[0,2]}$ &$0$ &$0$ &$0$   \\
\hline
$g^{\bullet}_{[0,2]}$ &$1$ &$1$ &$1$  \\
 \hline\end{tabular}
\end{center}
The classical configuration $g^{\circ}_{[0,2]}$ corresponds to 
 the Fock vector $\Omega_0$ (restricted to the local region 
$[0,1,2]$).  We  shall  write simply   $g^{\circ}_{[0,2]}=\Omega_0$, and 
 this identification will be used hereafter.
On the other hand, $g^{\bullet}_{[0,2]}$ corresponds 
to $\Qseq(\IDplzeroone)\Omega_0$ 
 which is the fully-occupied state on  $\Izeroone$,  where 
$\IDplzeroone\in \hat{\Xi}_{0,1}$. 
Thus we obtain   $g^{\bullet}_{[0,2]}=
\IDplzeroone\Omega_0$ which is  the   desired formula.

Second, let us  consider the case $n=2$.
$\widehat{\Upsilon}_{0,2}$ consists of the following 6 sequences 
on $\Izerotwo$:
\begin{center}
\begin{tabular}{|c||cc|c|cc|}
\hline
$\widehat{\Upsilon}_{0,2}$ &$0$ &$1$ &$2$ &$3$ &$4$ \\ \hline
$g^{\circ}_{[0,4]}$ &$0$ &$0$ &$0$ &$0$ &$0$  \\
\hline
$g_{[0,4]}^{1}$  &$0$ &$0$ &$0$ &$1$ &$1$  \\
$g_{[0,4]}^{2}$  &$0$ &$0$ &$1$ &$1$ &$1$  \\
\hline
$g_{[0,4]}^{3}$  &$1$ &$1$ &$1$ &$0$ &$0$  \\
$g_{[0,4]}^{4}$  &$1$ &$1$ &$0$ &$0$ &$0$  \\
\hline
$g^{\bullet}_{[0,4]}$ &$1$ &$1$ &$1$ &$1$ &$1$  \\
 \hline\end{tabular}
\end{center}
 We have  $g^{\circ}_{[0,4]}=\Omega_0$
 and $g^{\bullet}_{[0,4]}=\IDplzerotwo\Omega_0=
c^{\ast}_{0} c^{\ast}_{1} c^{\ast}_{2}c^{\ast}_{3} c^{\ast}_{4}
\Omega_0$ (the fully-occupied state on $\Izerotwo$) 
 according to \eqref{eq:Xi0-2}.
We have
\begin{align*}
g^{3}_{[0,4]}=\IDplzeroone\Omega_0=
c^{\ast}_{0} c^{\ast}_{1} c^{\ast}_{2}
\Omega_0,\ \ \IDplzeroone\in \hat{\Xi}_{0,1} \nonumber\\
g^{2}_{[0,4]}=\IDplonetwo\Omega_0=
c^{\ast}_{2} c^{\ast}_{3} c^{\ast}_{4}\Omega_0,
\ \ \IDplonetwo\in \hat{\Xi}_{1,2}. 
\end{align*}
To get  $g^{1}_{[0,4]}$ and $g^{4}_{[0,4]}$
 we  use `double'  applications  as 
\begin{align}
\label{eq:g104}
g^{1}_{[0,4]}=\IDmizeroone g^{\bullet}_{[0,4]}=
\IDmizeroone
\IDplzerotwo\Omega_0
= c^{\ast}_{3} c^{\ast}_{4}
\Omega_0, \quad \IDmizeroone\in \hat{\Xi}_{0,1}, 
\ \IDplzerotwo\in \hat{\Xi}_{0,2},
\end{align}
 and  similarly 
\begin{align}
\label{eq:g404}
g^{4}_{[0,4]}=\IDmionetwo g^{\bullet}_{[0,4]}=
\IDmionetwo
\IDplzerotwo\Omega_0
= c^{\ast}_{0} c^{\ast}_{1}
\Omega_0, \quad \IDmionetwo\in \hat{\Xi}_{1,2}, 
\ \IDplzerotwo\in \hat{\Xi}_{0,2}.
\end{align}
We have derived  all the elements  of  $\widehat{\Upsilon}_{0,2}$ and
 accordingly  
  all the  classical free-boundary supersymmetric ground states
 on $\Al(\Izerotwo)$ from  the Fock vector $\Omega_0$.
Let us note that $g^{1}_{[0,4]}$ and $g^{3}_{[0,4]}$
 are mapped to each other  by the particle-hole  transformation, and 
 so are  $g^{2}_{[0,4]}$ and $g^{4}_{[0,4]}$.
However, as shown  above, we do {\it{not}} need 
 to use the particle-hole  transformation. 

In an  analogous manner, we can get 
 all the elements  of  $\widehat{\Upsilon}_{0,2}$
 (all the  classical free-boundary supersymmetric ground states
 on $\Al(\Izerotwo)$) from  the fully-occupied  vector 
 $\Omega_1$ in place of $\Omega_0$.
 This fact  is important.
 \ \\

Let us consider  the case $n=3$.
$\widehat{\Upsilon}_{0,3}$ consists of the following 18  sequences 
on $\Izerothree$:
\begin{center}
\begin{tabular}{|c||cc|ccc|cc|}
\hline
$\widehat{\Upsilon}_{0,3}$ &$0$ &$1$ &$2$ &$3$ &$4$ &$5$ &$6$ \\ \hline
$g^{\circ}_{[0,6]} $ &$0$ &$0$ &$0$ &$0$ &$0$ &$0$ &$0$  \\
$g^{1}_{[0,6]}$ &$0$ &$0$ &$0$ &$1$ &$1$ &$0$ &$0$  \\
$g^{2}_{[0,6]}$ &$0$ &$0$ &$1$ &$1$ &$0$ &$0$ &$0$  \\
$g^{3}_{[0,6]}$ &$0$ &$0$ &$0$ &$1$ &$0$ &$0$ &$0$  \\
$g^{4}_{[0,6]}$ &$0$ &$0$ &$1$ &$1$ &$1$ &$0$ &$0$  \\
\hline
$g^{5}_{[0,6]}$  &$0$ &$0$ &$0$ &$0$ &$0$ &$1$ &$1$  \\
$g^{6}_{[0,6]}$   &$0$ &$0$ &$0$ &$0$ &$1$ &$1$ &$1$   \\
$g^{7}_{[0,6]}$   &$0$ &$0$ &$0$ &$1$ &$1$ &$1$ &$1$   \\
$g^{8}_{[0,6]}$   &$0$ &$0$ &$1$ &$1$ &$1$ &$1$ &$1$   \\
\hline
$g^{9}_{[0,6]}$   &$1$ &$1$ &$1$ &$1$ &$1$ &$0$ &$0$   \\
$g^{10}_{[0,6]}$   &$1$ &$1$ &$1$ &$1$ &$0$ &$0$ &$0$   \\
$g^{11}_{[0,6]}$   &$1$ &$1$ &$1$ &$0$ &$0$ &$0$ &$0$   \\
$g^{12}_{[0,6]}$   &$1$ &$1$ &$0$ &$0$ &$0$ &$0$ &$0$  \\
\hline
$g^{\bullet}_{[0,6]}$ &$1$ &$1$ &$1$ &$1$ &$1$ &$1$ &$1$  \\
$g^{13}_{[0,6]}$ &$1$ &$1$ &$1$ &$0$ &$0$ &$1$ &$1$  \\
$g^{14}_{[0,6]}$ &$1$ &$1$ &$0$ &$0$ &$1$ &$1$ &$1$  \\
$g^{15}_{[0,6]}$ &$1$ &$1$ &$1$ &$0$ &$1$ &$1$ &$1$  \\
$g^{16}_{[0,6]}$ &$1$ &$1$ &$0$ &$0$ &$0$ &$1$ &$1$  \\
 \hline
\end{tabular}
\end{center}

We will  generate all the above   
$g^{\ast}_{[0,6]}$. 
Obviously $g^{\circ}_{[0,6]}=\Omega_0$, 
 and $g^{\bullet}_{[0,6]}=\IDplzerothree\Omega_0=
c^{\ast}_{0} c^{\ast}_{1} c^{\ast}_{2}c^{\ast}_{3} c^{\ast}_{4}
c^{\ast}_{5} c^{\ast}_{6}\Omega_0$ which is the fully-occupied 
 vector  $\Omega_1$
restricted to  $\Izerothree$.

The restriction of $g^{1}_{[0,6]}$ to $[0,4]$ is
 $g^{1}_{[0,4]}$,     
and the restriction of $g^{4}_{[0,6]}$ to $[0,4]$ is
 $g^{2}_{[0,4]}$.
The restriction of $g^{2}_{[0,6]}$ to $[2,6]$ is
 $g^{4}_{[2,6]}$,  which is the translation   of 
 $g^{4}_{[0,4]}$ used before.
Thus each  of $g^{1}_{[0,6]}$
 $g^{4}_{[0,6]}$ and $g^{2}_{[0,6]}$
 can be given as in  the case $n=2$.

The restriction of $g^{13}_{[0,6]}$ to $[0,4]$ is
 $g^{3}_{[0,4]}$,     
 the restriction of $g^{14}_{[0,6]}$ to $[2,6]$ is
 $g^{2}_{[2,6]}$, 
 the restriction of $g^{16}_{[0,6]}$ to $[0,4]$ is
 $g^{4}_{[0,4]}$ (also  the restriction of $g^{16}_{[0,6]}$ to $[2,6]$ is
 $g^{4}_{[2,6]}$).
Hence  each  of $g^{13}_{[0,6]}$,
 $g^{14}_{[0,6]}$ and $g^{16}_{[0,6]}$
 can be given as in the case $n=2$.
Note that   
 $\Omega_1|_{[0,6]}=\IDplzerothree\Omega_0|_{[0,6]}$
 with $\IDplzerothree\in \Al(\Izerothree)$.
Therefore  each  of $g^{13}_{[0,6]}$
 $g^{14}_{[0,6]}$ and $g^{16}_{[0,6]}$
 can be generated  by local supercharges  in $[0,6]$ 
applied to $\Omega_0$.

We have 
\begin{align*}
g^{3}_{[0,6]}=\IDmizeroone \IDmitwothree
\Omega_1=
\IDmizeroone \IDmitwothree
\IDplzerothree\Omega_0= c^{\ast}_{3}\Omega_0,
\end{align*}
and 
\begin{align*}
g^{15}_{[0,6]}=\IDplzeroone \IDpltwothree
\Omega_0=c^{\ast}_{0}c^{\ast}_{1} c^{\ast}_{2}
c^{\ast}_{4}c^{\ast}_{5} c^{\ast}_{6}\Omega_0.
\end{align*}

We have 
\begin{align*}
g^{5}_{[0,6]}&=\IDmizerotwo
\Omega_1=
 \IDmizerotwo
\IDplzerothree\Omega_0= c^{\ast}_{5}c^{\ast}_{6}
\Omega_0, \nonumber\\
g^{6}_{[0,6]}&=
\IDpltwothree\Omega_0=c^{\ast}_{4}c^{\ast}_{5}c^{\ast}_{6}\Omega_0,\nonumber\\
g^{7}_{[0,6]}&=
\IDmizeroone
\Omega_1=
\IDmizeroone\IDplzerothree{\Omega_0}
=c^{\ast}_{3} c^{\ast}_{4}c^{\ast}_{5}c^{\ast}_{6}\Omega_0,\nonumber\\
g^{8}_{[0,6]}&=
\IDplonethree{\Omega_0}=
c^{\ast}_{2}c^{\ast}_{3}
c^{\ast}_{4}c^{\ast}_{5}c^{\ast}_{6}\Omega_0,\nonumber\\
\end{align*}
 and  similarly 
\begin{align*}
g^{9}_{[0,6]}&=\IDplzerotwo\Omega_0=
c^{\ast}_{0}
c^{\ast}_{1}c^{\ast}_{2} c^{\ast}_{3}c^{\ast}_{4}
\Omega_0, \nonumber\\
g^{10}_{[0,6]}&=
\IDmitwothree
\Omega_1=
 \IDmitwothree
\IDplzerothree\Omega_0= 
c^{\ast}_{0}
c^{\ast}_{1}c^{\ast}_{2} c^{\ast}_{3}
\Omega_0, \nonumber\\
g^{11}_{[0,6]}&=
\IDplzeroone{\Omega_0}=
c^{\ast}_{0}
c^{\ast}_{1}c^{\ast}_{2}
\Omega_0,\nonumber\\
g^{12}_{[0,6]}&=
\IDmionethree
\Omega_1=
 \IDmionethree
\IDplzerothree \Omega_0= 
c^{\ast}_{0}c^{\ast}_{1}\Omega_0.
\end{align*}

We have now derived  all the sequences of  $\widehat{\Upsilon}_{0,3}$, i.e. 
  all the  classical free-boundary supersymmetric ground states
 on $\Al(\Izerothree)$.

\vspace*{1cm}

We will start the argument of induction.
We have  verified  the statement for $n=1,2,3$.
Now let us assume that the statement holds for any 
 integer from $1\in\NN$ up to $n\in\NN$.
We are going to  show  that the statement holds  for $n+1\in\NN$.
Concretely,   we will   construct $\widehat{\Upsilon}_{0,n+1}$
 from  $\widehat{\Upsilon}_{p,q}$  ($0\le p <q \le n+1$) 
where $0<p$ or $q<n+1$.

We  divide   $\widehat{\Upsilon}_{0,n+1}$ into four cases  (Case I-IV) as below. We shall indicate  how the induction argument can be applied to each
 of them.\\

\noindent Case I:\\
We deal with all $g\in \widehat{\Upsilon}_{0,n+1}$
 whose  left and right  ends  are  
\begin{align}
\label{eq:caseI}
g(0)=g(1)=0,\ \ g(2n+1)=g(2(n+1))=0. 
\end{align}

\begin{center}
\begin{tabular}{|c||cc|cc|ccc|cc|cc|}
\hline
$\widehat{\Upsilon}_{0,n+1}$ &
$0$ &$1$ &$2$ &$3$ &$ \cdots$ &$ \cdots$ &$\cdots $ &$2n-1$ &$2n$ &$2n+1$ &$2(n+1)$ 
\\ \hline
I-1& $0$ &$0$ &$0$ &$0$ &$*** $ &$***$ &$***$ &$0$ &$0$ &$0$ &$0$ \\
\hline
I-2&$0$ &$0$ &$0$ &$0$ &$*** $ &$***$ &$***$ &$1$ &$1$ &$0$ &$0$ \\
\hline
I-3&$0$ &$0$ &$0$ &$0$ &$*** $ &$***$ &$***$ &$1$ &$0$ &$0$ &$0$ \\
\hline
I-4&$0$ &$0$ &$1$ &$1$ &$*** $ &$***$ &$***$ &$0$ &$0$ &$0$ &$0$ \\
\hline
I-5&$0$ &$0$ &$1$ &$1$ &$*** $ &$***$ &$***$ &$1$ &$1$ &$0$ &$0$ \\
\hline
I-6&$0$ &$0$ &$1$ &$1$ &$*** $ &$***$ &$***$ &$1$ &$0$ &$0$ &$0$ \\
\hline
I-7&$0$ &$0$ &$0$ &$1$ &$*** $ &$***$ &$***$ &$0$ &$0$ &$0$ &$0$ \\
\hline
I-8&$0$ &$0$ &$0$ &$1$ &$*** $ &$***$ &$***$ &$1$ &$1$ &$0$ &$0$ \\
\hline
I-9&$0$ &$0$ &$0$ &$1$ &$*** $ &$***$ &$***$ &$1$ &$0$ &$0$ &$0$ \\
\hline
\end{tabular}
\end{center}

{Note that `$***$'s in the middle mean some appropriate 
sequences of $0,1$ so that 
 the sequence belongs to $\widehat{\Upsilon}_{0,n+1}$, not being arbitrary.}

All the above elements 
 in  $\widehat{\Upsilon}_{0,n+1}$ except  I-9 
 belong to $\widehat{\Upsilon}_{1,n+1}$
 or to $\widehat{\Upsilon}_{0,n}$
 when being restricted to $[2, 2(n+1)]$ or 
to $[0, 2n]$, respectively.
By applying  $\IDplzeroone$ to the vector of I-9, we get   
\begin{center}
\begin{tabular}{|c||cc|cc|ccc|cc|cc|}
\hline
$\widehat{\Upsilon}_{0,n+1}$ &
$0$ &$1$ &$2$ &$3$ &$ \cdots$ &$ \cdots$ &$\cdots $ &$2n-1$ &$2n$ &$2n+1$ &$2(n+1)$ 
\\ \hline
 New I-9&$1$ &$1$ &$1$ &$1$ &$*** $ &$***$ &$***$ &$1$ &$0$ &$0$ &$0$ \\
\hline
\end{tabular}
\end{center}
``New I-9''  above 
 belongs to 
  $\widehat{\Upsilon}_{1,n+1}$
 when being restricted to $[2, 2(n+1)]$. 
Therefore we can obtain New I-9
by applying  some local supercharges in  
$[2, 2(n+1)]$ to  $\Omega_1$ (not $\Omega_0$ here).
Note that  $\Omega_1=\IDplzeronone\Omega_0$
 on the segment $[0, 2(n+1)]$ as noted in Lemma \ref{lem:vacuum-occupied-start}.   Hence  we can construct 
  I-9 by applying  some local supercharges in  
$[0, 2(n+1)]$ to   $\Omega_0$.
In this way we have  made all the configurations 
  of  Case I by the specified rule.

By applying  $\IDplnnone$ to the vector of  I-9, we get
\begin{center}
\begin{tabular}{|c||cc|cc|ccc|cc|cc|}
\hline
$\widehat{\Upsilon}_{0,n+1}$ &
$0$ &$1$ &$2$ &$3$ &$ \cdots$ &$ \cdots$ &$\cdots $ &$2n-1$ &$2n$ &$2n+1$ &$2(n+1)$ 
\\ \hline
 New I-9(2)&$0$ &$0$ &$0$ &$1$ &$*** $ &$***$ &$***$ &$1$ &$1$ &$1$ &$1$ \\
\hline
\end{tabular}
\end{center}
``New I-9(2)''  above 
 belongs to 
  $\widehat{\Upsilon}_{0,n}$
 when being restricted to $[0, 2n]$. 
Therefore we can obtain New I-9(2)
by applying  some local supercharges in  
$[0, 2n]$ to  $\Omega_1$ (not $\Omega_0$ here).
By noting  Lemma \ref{lem:vacuum-occupied-start}   we can construct 
I-9 by applying  some local supercharges in  
$[0, 2(n+1)]$ to  $\Omega_0$.

\ \\

\noindent Case II:\\
We deal with all $g\in \widehat{\Upsilon}_{0,n+1}$
 whose  left and right  ends  are  
\begin{align}
\label{eq:caseII}
g(0)=g(1)=1,\ \ g(2n+1)=g(2(n+1))=1. 
\end{align}
The proof for Case II can be done in the same way 
 as done for Case I.\\

\noindent Case III:\\
We deal with all $g\in \widehat{\Upsilon}_{0,n+1}$
 whose  left and right  ends  are  
\begin{align}
\label{eq:caseIII}
g(0)=g(1)=0,\ \ g(2n+1)=g(2(n+1))=1. 
\end{align}

\begin{center}
\begin{tabular}{|c||cc|cc|ccc|cc|cc|}
\hline
$\widehat{\Upsilon}_{0,n+1}$ &
$0$ &$1$ &$2$ &$3$ &$ \cdots$ &$ \cdots$ &$\cdots $ &$2n-1$ &$2n$ &$2n+1$ &$2(n+1)$ 
\\ \hline
III-1& $0$ &$0$ &$0$ &$0$ &$*** $ &$***$ &$***$ &$0$ &$0$ &$1$ &$1$ \\
\hline
III-2&$0$ &$0$ &$0$ &$0$ &$*** $ &$***$ &$***$ &$1$ &$1$ &$1$ &$1$ \\
\hline
III-3&$0$ &$0$ &$0$ &$0$ &$*** $ &$***$ &$***$ &$0$ &$1$ &$1$ &$1$ \\
\hline
III-4&$0$ &$0$ &$1$ &$1$ &$*** $ &$***$ &$***$ &$0$ &$0$ &$1$ &$1$ \\
\hline
III-5&$0$ &$0$ &$1$ &$1$ &$*** $ &$***$ &$***$ &$1$ &$1$ &$1$ &$1$ \\
\hline
III-6&$0$ &$0$ &$1$ &$1$ &$*** $ &$***$ &$***$ &$0$ &$1$ &$1$ &$1$ \\
\hline
III-7&$0$ &$0$ &$0$ &$1$ &$*** $ &$***$ &$***$ &$0$ &$0$ &$1$ &$1$ \\
\hline
III-8&$0$ &$0$ &$0$ &$1$ &$*** $ &$***$ &$***$ &$1$ &$1$ &$1$ &$1$ \\
\hline
III-9&$0$ &$0$ &$0$ &$1$ &$*** $ &$***$ &$***$ &$0$ &$1$ &$1$ &$1$ \\
\hline
\end{tabular}
\end{center}

All the above elements 
 in  $\widehat{\Upsilon}_{0,n+1}$ except  III-9 
 belong to $\widehat{\Upsilon}_{1,n+1}$
 or to $\widehat{\Upsilon}_{0,n}$
 when being restricted to $[2, 2(n+1)]$ or 
to $[0, 2n]$, respectively.
By applying  $\IDplzeroone$ to the vector of  III-9, we get   
\begin{center}
\begin{tabular}{|c||cc|cc|ccc|cc|cc|}
\hline
$\widehat{\Upsilon}_{0,n+1}$ &
$0$ &$1$ &$2$ &$3$ &$ \cdots$ &$ \cdots$ &$\cdots $ &$2n-1$ &$2n$ &$2n+1$ &$2(n+1)$ 
\\ \hline
 New III-9&$1$ &$1$ &$1$ &$1$ &$*** $ &$***$ &$***$ &$0$ &$1$ &$1$ &$1$ \\
\hline
\end{tabular}
\end{center}

``New III-9''  above 
 belongs to   $\widehat{\Upsilon}_{1,n+1}$
 when being restricted to $[2, 2(n+1)]$. 
Therefore we can obtain New III-9
by applying  some local supercharges in  
$[2, 2(n+1)]$ to $\Omega_1$ (not $\Omega_0$ here).
Note that  $\Omega_1$ can be constructed from  $\Omega_0$
 on $[0, 2(n+1)]$ by using  local supercharges on $[0, 2(n+1)]$. 
  Thus we can construct 
I-9 by applying  some local supercharges in  $[0, 2(n+1)]$ to  $\Omega_0$.
We have  completed   the assertion for Case III.\\ 
 
\noindent Case IV:\\
We deal with all $g\in \widehat{\Upsilon}_{0,n+1}$
 whose  left and right  ends  are  
\begin{align}
\label{eq:caseII}
g(0)=g(1)=1,\ \ g(2n+1)=g(2(n+1))=0. 
\end{align}
The proof for Case IV is similar to that  for Case III given above.\\

In conclusion, 
for all the cases (Case I-IV)
we have generated all the elements of 
$\widehat{\Upsilon}_{0,n+1}$ 
 from $\widehat{\Upsilon}_{0,n}$ and 
 $\widehat{\Upsilon}_{1,n+1}$.
 Hence by the induction, we have shown the statement. 
\end{proof}

The number of classical supersymmetric ground states 
can be computed explicitly.
\begin{prop}
\label{prop:transfer}
The number of classical free-boundary supersymmetric ground states 
 on $\Izeron$ ($n\in \NN$) 
 is  $2\cdot3^{n-1}$. 
\end{prop}

\begin{proof}
This  computation  is given by the transfer-matrix method. 
We first divide  $\Izeron$
into $n$-sequential pairs as
\begin{equation*} 
\Izeron = \{0, 1,2\}\cup \{3,4\}  \cdots \cup \{2k-1,2k\}\cup \cdots \cup \{2n-1,2n\},
\end{equation*}
 where  the first group exceptionally   consists of three sites  $\{0, 1,2\}$.
On each $\{2k-1, 2k\}$
 all  classical configurations are  possible.
However, to connect $\{2k-1,2k\}$ and
 $\{2k+1,2(k+1)\}$ we have to avoid the forbidden  triplets:
$\{0,1,0\}$  $\{1,0,1\}$ on $\{2k-1,2k, 2k+1\}$.
So the transfer matrix should be
\begin{equation} 
\label{eq:Trans}
T:=
\begin{bmatrix} 
1 & 1&1&1 \\ 
0 & 0&1&1 \\ 
1 & 1&0&0 \\ 
1 & 1&1&1 \\ 
\end{bmatrix}. 
\end{equation} 
By taking the edge condition 
\eqref{eq:edge-const} into account, 
the possible configurations are 
any of 
\begin{align*} 
0,0,0, \ldots ,0,0,\\
0,0,0, \ldots ,1,1,\\
0,0,1, \ldots ,0,0,\\
0,0,1, \ldots ,1,1,\\
1,1,0, \ldots ,0,0,\\
1,1,0, \ldots ,1,1,\\
1,1,1, \ldots ,0,0,\\
1,1,1, \ldots ,1,1,\\
\end{align*}
which correspond to 
$(1,1)$, $(1,4)$,
 $(2,1)$, $(2,4)$, $(3,1)$, $(3,4)$,
 $(4,1)$ and $(4,4)$ elements 
 of $T^{n-1}$, respectively.
 Those   amount to  $2\cdot 3^{n-1}$.
\end{proof}

\section{{Discussion}}
\label{sec:dis}
We  determined  
all {\it{classical}} supersymmetric ground states 
 of the Nicolai supersymmetric fermion lattice model, 
 and  explained   the  high degeneracy of ground states  
 by  breakdown of its  infinitely many  local fermionic symmetries.
The above  finding  
 may  recall other supersymmetric models with many ground states 
  such as  the supersymmetric fermion lattice model
  by Fendley-Schoutens-de Boer-Nienhuis \cite{FEN1} \cite{FEN2}  
on two-dimensional lattice \cite{vanEER}, and some   Wess-Zumino supersymmetry  quantum mechanical model  \cite{ARAI}.

In  
\cite{RUBEN}  the exact 
 number of ground states on finite systems of the Nicolai model is shown, 
 but  the precise form of these states    is  not specified.
To determine all the  ground states (furthermore all eigenstates) 
beyond the classical ground  states discussed  in this paper 
we need  more detailed spectral property 
 of the Hamiltonian and its symmetries (including bosonic ones).

 In \cite{SAN} \cite{MORIYA-prd}
  an extended version of Nicolai model
 that breaks its dynamical supersymmetry is  studied.
As we have seen, the (original) Nicolai model   
 does not break its dynamical supersymmetry.   
However,  it will break its  hidden 
 supersymmetries for some ground states. 
It would be interesting to discuss breakdown 
 of these  hidden  fermionic  symmetries.

The final comment is concerned with some technical point.
We have chosen special 
  subregions ($\Ikl$) and the boundary conditions 
(the free-boundary supersymmetric condition) which  
 seem  artificial. 
However, as long as we consider classical states only,  
 there is  no loss of  generality with this choice as follows. 
\begin{prop}
\label{prop:Iklenough}
Given any finite subset  $\Lambda$ of $\Z$. 
Any classical  supersymmetric  state on $\Lambda$ can be given by restriction 
 of  some classical  free-boundary supersymmetric ground  state on some larger $\Ikl$
  that includes  $\Lambda$. 
\end{prop}

\begin{proof}
First  recall the  one-to-one correspondence  between  
 the set of  classical supersymmetric ground states on $\Lambda$ 
 and ${\Upsilon}_{\Lambda}$  by Proposition  
\ref{prop:AllclassicZ}.
Recall  the  one-to-one correspondence  between  
 the set of  classical free-boundary supersymmetric ground states 
on $\Ikl$ 
 and $\widehat{\Upsilon}_{k,l}$  by Proposition \ref{prop:classical-Ikl}.
Thus any  classical supersymmetric ground state on  $\Lambda$ 
can be extended  to  
at least one  classical  free-boundary supersymmetric ground state on 
 $\Ikl$ that includes $\Lambda$.
\end{proof}

\section*{Acknowledgments}
H. K. was supported in part by JSPS Grant-in-Aid for  Scientic Research on Innovative Areas No. JP18H04478 and JP20H04630, and JSPS KAKENHI Grant No. 18K03445. 
H. M.   would like to thank   Prof. Arai and 
 Dr. Huijse for helpful discussion. H. M. acknowledges 
  Riyu-1  group of Kanazawa University for encouragement.  


\appendix
\section{Appendix}

\subsection{Forms of $\hat{\Xi}$}
\label{subsec:example-hatXi}
We will give concrete examples  for 
 local  $\{-1, +1\}$-sequences of conservation 
   of   Definition \ref{defn:seq-conservation} and 
  their associated   local fermion operators 
   of   Definition \ref{defn:local-assignment}.
First we  see that  $\hat{\Xi}_{0,1}$  on $\Izeroone\equiv[0,1,2]$
  consists of  two  $\pm$-characters only.

\begin{center}
\begin{tabular}{|c||ccc|}
\hline
${\hat{\Xi}_{0,1}}$ &$0$ &$1$ &$2$  \\ \hline
$\IDmizeroone$ &$-1$ &$-1$ &$-1$   \\
\hline
$\IDplzeroone$ &$+1$ &$+1$ &$+1$  \\
 \hline\end{tabular}
\end{center}
By  \eqref{eq:Qseq}  of Definition \ref{defn:local-assignment}
 the corresponding local fermion operators are  
\begin{align}
\label{eq:Xi0-1}
\Qseq(\IDmizeroone)&=c_{0} c_{1} c_{2}
\in {\Al}(\Izeroone)_{-},\nonumber\\
\Qseq(\IDplzeroone)&=c^{\ast}_{0} c^{\ast}_{1} c^{\ast}_{2}
\in {\Al}(\Izeroone)_{-}.
\end{align}

We consider 
   a next  smallest segment  $\Izerotwo\equiv[0,1,2,3,4]$ 
  by setting  $k=0$ and $l=2$. 
The space 
$\hat{\Xi}_{0,2}$ on $\Izerotwo$
 consists of the following five  $\{-1, +1\}$-sequences:
\begin{center}
\begin{tabular}{|c||cc|c|cc|}
\hline
${\hat{\Xi}_{0,2}}$ &$0$ &$1$ &$2$ &$3$ &$4$ \\ \hline
$\IDmizerotwo$ &$-1$ &$-1$ &$-1$ &$-1$ &$-1$  \\
\hline
$u_{[0,4]}^{\rm{i}}$  &$-1$ &$-1$ &$-1$ &$+1$ &$+1$  \\
$u_{[0,4]}^{\rm{ii}}$  &$-1$ &$-1$ &$+1$ &$+1$ &$+1$  \\
\hline
$v_{[0,4]}^{\rm{i}}$  &$+1$ &$+1$ &$+1$ &$-1$ &$-1$  \\
$v_{[0,4]}^{\rm{ii}}$  &$+1$ &$+1$ &$-1$ &$-1$ &$-1$  \\
\hline
$\IDplzerotwo$ &$+1$ &$+1$ &$+1$ &$+1$ &$+1$  \\
 \hline\end{tabular}
\end{center}
Note that 
\begin{align}
\label{eq:Xi0-2-minus}
\IDmizerotwo=-\IDplzerotwo,\ u_{[0,4]}^{\rm{i}}=-v_{[0,4]}^{\rm{i}},\ 
u_{[0,4]}^{\rm{ii}}=-v_{[0,4]}^{\rm{ii}}.
\end{align}
By \eqref{eq:Qseq}  of Definition \ref{defn:local-assignment}
we have 
\begin{align}
\label{eq:Xi0-2}
\Qseq(\IDmizerotwo)&=c_{0} c_{1} c_{2}c_{3} c_{4}
\in {\Al}(\Izerotwo)_{-},\nonumber\\
\Qseq(u_{[0,4]}^{\rm{i}})&=c_{0} c_{1} c_{2}c^{\ast}_{3} c^{\ast}_{4}\in {\Al}
(\Izerotwo)_{-},\nonumber\\
\Qseq(u_{[0,4]}^{\rm{ii}})&=c_{0} c_{1} c^{\ast}_{2}c^{\ast}_{3} c^{\ast}_{4}
\in {\Al}(\Izerotwo)_{-},\nonumber\\
\Qseq(v_{[0,4]}^{\rm{i}})&=c^{\ast}_{0} c^{\ast}_{1} c^{\ast}_{2}c_{3} c_{4}
\in {\Al}(\Izerotwo)_{-},\nonumber\\
\Qseq(v_{[0,4]}^{\rm{ii}})&=c^{\ast}_{0} c^{\ast}_{1} c_{2}c_{3}c_{4}
\in {\Al}(\Izerotwo)_{-},\nonumber\\
\Qseq(\IDplzerotwo)&=c^{\ast}_{0} c^{\ast}_{1} c^{\ast}_{2}c^{\ast}_{3} c^{\ast}_{4}
\in {\Al}(\Izerotwo)_{-}.
\end{align}

We then consider    the segment  $\Izerothree\equiv[0,1,2,3,4,5,6]$ 
  taking   $k=0$ and $l=3$. 
By definition it consists  of  $5+4+4+5=18$ $\{-1, +1\}$-sequences:
\begin{center}
\begin{tabular}{|c||cc|ccc|cc|}
\hline
${\hat{\Xi}_{0,3}}$ &$0$ &$1$ &$2$ &$3$ &$4$ &$5$ &$6$ \\ \hline
$s_{[0,6]}^{\circ} $ &$-1$ &$-1$ &$-1$ &$-1$ &$-1$ &$-1$ &$-1$  \\
$s_{[0,6]}^{\rm{i}}$ &$-1$ &$-1$ &$-1$ &$+1$ &$+1$ &$-1$ &$-1$  \\
$s_{[0,6]}^{\rm{ii}}$ &$-1$ &$-1$ &$+1$ &$+1$ &$-1$ &$-1$ &$-1$  \\
$s_{[0,6]}^{\rm{iii}}$ &$-1$ &$-1$ &$-1$ &$+1$ &$-1$ &$-1$ &$-1$  \\
$s_{[0,6]}^{\rm{iv}}$ &$-1$ &$-1$ &$+1$ &$+1$ &$+1$ &$-1$ &$-1$  \\
\hline
$u_{[0,6]}^{\rm{i}}$  &$-1$ &$-1$ &$-1$ &$-1$ &$-1$ &$+1$ &$+1$  \\
$u_{[0,6]}^{\rm{ii}}$   &$-1$ &$-1$ &$-1$ &$-1$ &$+1$ &$+1$ &$+1$   \\
$u_{[0,6]}^{\rm{iii}}$   &$-1$ &$-1$ &$-1$ &$+1$ &$+1$ &$+1$ &$+1$   \\
$u_{[0,6]}^{\rm{iv}}$   &$-1$ &$-1$ &$+1$ &$+1$ &$+1$ &$+1$ &$+1$   \\
\hline
$v_{[0,6]}^{\rm{i}}$   &$+1$ &$+1$ &$+1$ &$+1$ &$+1$ &$-1$ &$-1$   \\
$v_{[0,6]}^{\rm{ii}}$   &$+1$ &$+1$ &$+1$ &$+1$ &$-1$ &$-1$ &$-1$   \\
$v_{[0,6]}^{\rm{iii}}$   &$+1$ &$+1$ &$+1$ &$-1$ &$-1$ &$-1$ &$-1$   \\
$v_{[0,6]}^{\rm{iv}}$   &$+1$ &$+1$ &$-1$ &$-1$ &$-1$ &$-1$ &$-1$  \\
\hline
$t_{[0,6]}^{\bullet}$ &$+1$ &$+1$ &$+1$ &$+1$ &$+1$ &$+1$ &$+1$  \\
$t_{[0,6]}^{\rm{i}}$ &$+1$ &$+1$ &$+1$ &$-1$ &$-1$ &$+1$ &$+1$  \\
$t_{[0,6]}^{\rm{ii}}$ &$+1$ &$+1$ &$-1$ &$-1$ &$+1$ &$+1$ &$+1$  \\
$t_{[0,6]}^{\rm{iii}}$ &$+1$ &$+1$ &$+1$ &$-1$ &$+1$ &$+1$ &$+1$  \\
$t_{[0,6]}^{\rm{iv}}$ &$+1$ &$+1$ &$-1$ &$-1$ &$-1$ &$+1$ &$+1$  \\
 \hline\end{tabular}
\end{center}
Note that 
$s_{[0,6]}^{\circ}\equiv \IDmizerothree$ and 
$t_{[0,6]}^{\bullet}\equiv \IDplzerothree$
and that 
\begin{align}
\label{eq:Xi0-3-minus}
s_{[0,6]}^{\circ}=-t_{[0,6]}^{\bullet},\ 
s_{[0,6]}^{k}=-t_{[0,6]}^{k},\forall k\in \{\rm{i}, \rm{ii}, \rm{iii}, \rm{iv}\} \nonumber\\
u_{[0,6]}^{k}=-v_{[0,6]}^{k}, 
\forall k\in \{\rm{i}, \rm{ii}, \rm{iii}, \rm{iv}\}.
\end{align}
Using  the rule 
we obtain  the following list of 18 fermion operators associated to
 ${\hat{\Xi}_{0,3}}$:
\begin{align}
\label{eq:Xi0-3}
\Qseq(s_{[0,6]}^{\circ})&=c_{0} c_{1} c_{2}c_{3} c_{4} c_{5} c_{6}
\in {\Al}(\Izerothree)_{-},\nonumber\\
\Qseq(s_{[0,6]}^{\rm{i}})&=c_{0} c_{1} c_{2}c^{\ast}_{3} c^{\ast}_{4} c_{5} c_{6}
\in {\Al}(\Izerothree)_{-},\nonumber\\
\Qseq(s_{[0,6]}^{\rm{ii}})&=c_{0} c_{1} c^{\ast}_{2}c^{\ast}_{3} c_{4} c_{5} c_{6}
\in {\Al}(\Izerothree)_{-},\nonumber\\
\Qseq(s_{[0,6]}^{\rm{iii}})&=c_{0} c_{1} c_{2}c^{\ast}_{3} c_{4} c_{5} c_{6}
\in {\Al}(\Izerothree)_{-},\nonumber\\
\Qseq(s_{[0,6]}^{\rm{iv}})&=c_{0} c_{1} c^{\ast}_{2}c^{\ast}_{3} c^{\ast}_{4} c_{5} c_{6}
\in {\Al}(\Izerothree)_{-},\nonumber\\
\Qseq(u_{[0,6]}^{\rm{i}})&=c_{0} c_{1} c_{2}c_{3} c_{4} c^{\ast}_{5} c^{\ast}_{6}
\in {\Al}(\Izerothree)_{-},\nonumber\\
\Qseq(u_{[0,6]}^{\rm{ii}})&=c_{0} c_{1} c_{2}c_{3} c^{\ast}_{4} c^{\ast}_{5} c^{\ast}_{6}
\in {\Al}(\Izerothree)_{-},\nonumber\\
\Qseq(u_{[0,6]}^{\rm{iii}})&=c_{0} c_{1} c_{2}c^{\ast}_{3} c^{\ast}_{4} c^{\ast}_{5} c^{\ast}_{6}
\in {\Al}(\Izerothree)_{-},\nonumber\\
\Qseq(u_{[0,6]}^{\rm{iv}})&=c_{0} c_{1} c^{\ast}_{2}c^{\ast}_{3} c^{\ast}_{4} c^{\ast}_{5} c^{\ast}_{6}\in {\Al}(\Izerothree)_{-},\nonumber\\
\Qseq(v_{[0,6]}^{\rm{i}})&=c^{\ast}_{0} c^{\ast}_{1} c^{\ast}_{2}c^{\ast}_{3} 
c^{\ast}_{4} c_{5} c_{6}
\in {\Al}(\Izerothree)_{-},\nonumber\\
\Qseq(v_{[0,6]}^{\rm{ii}})&=c^{\ast}_{0} c^{\ast}_{1} c^{\ast}_{2}c^{\ast}_{3} c_{4} c_{5} c_{6}
\in {\Al}(\Izerothree)_{-},\nonumber\\
\Qseq(v_{[0,6]}^{\rm{iii}})&=c^{\ast}_{0} c^{\ast}_{1} c^{\ast}_{2}c_{3} c_{4} c_{5} c_{6}
\in {\Al}(\Izerothree)_{-},\nonumber\\
\Qseq(v_{[0,6]}^{\rm{iv}})&=c^{\ast}_{0} c^{\ast}_{1} c_{2}c_{3} 
c_{4} c_{5} c_{6}
\in {\Al}(\Izerothree)_{-},\nonumber\\
\Qseq(t_{[0,6]}^{\bullet})&=c^{\ast}_{0} c^{\ast}_{1} c^{\ast}_{2}c^{\ast}_{3} c^{\ast}_{4}
c^{\ast}_{5} c^{\ast}_{6}
\in {\Al}(\Izerothree)_{-},\nonumber\\
\Qseq(t_{[0,6]}^{\rm{i}})&=c^{\ast}_{0} c^{\ast}_{1} c^{\ast}_{2}c_{3} c_{4}
c^{\ast}_{5} c^{\ast}_{6}
\in {\Al}(\Izerothree)_{-},\nonumber\\
\Qseq(t_{[0,6]}^{\rm{ii}})&=c^{\ast}_{0} c^{\ast}_{1} c_{2}c_{3} c^{\ast}_{4}
c^{\ast}_{5} c^{\ast}_{6}
\in {\Al}(\Izerothree)_{-},\nonumber\\
\Qseq(t_{[0,6]}^{\rm{iii}})&=c^{\ast}_{0} c^{\ast}_{1} c^{\ast}_{2}c_{3} c^{\ast}_{4}
c^{\ast}_{5} c^{\ast}_{6}
\in {\Al}(\Izerothree)_{-},\nonumber\\
\Qseq(t_{[0,6]}^{\rm{iv}})&=c^{\ast}_{0} c^{\ast}_{1} c_{2}c_{3} c_{4}
c^{\ast}_{5} c^{\ast}_{6}
\in {\Al}(\Izerothree)_{-}.
\end{align}

\end{document}